\documentclass[11pt]{article}
\usepackage{times}
\usepackage{latexsym}
\usepackage{amsmath,amssymb,amsthm}

\usepackage{fullpage}
\newenvironment{myitemize}{\begin{itemize}}{\end{itemize}}

\usepackage{tikz}

\newtheorem{theorem}{Theorem}
\newtheorem{lemma}[theorem]{Lemma}

\newtheorem{fact}[theorem]{Fact}

\newtheorem{property}[theorem]{Property}

\newcommand{\Ot}{\widetilde{O}}
\newcommand{\Ost}{O^*}
\newcommand{\eps}{\varepsilon}
\newcommand{\calC}{{\cal C}}
\newcommand{\TO}{,\ldots,}
\newcommand{\R}{\mbox{\rm I\hspace{-0.2em}R}}
\newcommand{\up}[1]{{\left\lceil{#1}\right\rceil}}

\newcommand{\IGNORE}[1]{}

\def\compactify{\itemsep=0pt \topsep=0pt \partopsep=0pt \parsep=0pt}
\let\latexusecounter=\usecounter
\newenvironment{itemize*}
  {\def\usecounter{\compactify\latexusecounter}
   \begin{itemize}}
  {\end{itemize}\let\usecounter=\latexusecounter}
\newenvironment{enumerate*}
  {\def\usecounter{\compactify\latexusecounter}
   \begin{enumerate}}
  {\end{enumerate}\let\usecounter=\latexusecounter}

\title{Dynamic Connectivity: Connecting to Networks and Geometry}
\author{
     Timothy M. Chan  \\[2pt] University of Waterloo
\\ {\tt\small tmchan@uwaterloo.ca}
\and Mihai P\v{a}tra\c{s}cu \\[2pt] MIT
\\ {\tt\small mip@mit.edu} 
\and Liam Roditty \\[2pt] Weizmann Institute
\\ {\tt\small liam.roditty@weizmann.ac.il}
}

\begin{document}

\maketitle

\begin{abstract}
Dynamic connectivity is a well-studied problem, but so far
the most compelling progress has been confined to the edge-update
model: maintain an understanding of connectivity in an undirected
graph, subject to edge insertions and deletions. In this paper, we
study two more challenging, yet equally fundamental problems:
%models, in which existing results which seem more challenging.

\smallskip

{\bf Subgraph connectivity} asks to maintain an understanding of
connectivity under vertex updates: updates can turn vertices on and off,
and queries refer to the subgraph induced by \emph{on} vertices.  (For
instance, this is closer to applications in networks of routers, where
node faults may occur.)
%where a 
%misconfiguration or a crash-reboot are relatively frequent ``cheap''
%events, compared to accidental cable cuts.

We describe a data structure supporting vertex updates in
$\Ot(m^{2/3})$ amortized time, where $m$ denotes the number
of edges in the graph. This greatly improves over
the previous result {\rm [Chan, STOC'02]}, which
required fast matrix multiplication and had an update time 
of $O(m^{0.94})$.  The new data structure is also simpler.
%is a notable progress over the previous
%state of the art, both numerically and philosophically: the only known

\smallskip

{\bf Geometric connectivity} asks to maintain a dynamic set of $n$
geometric objects, and query connectivity in their intersection
graph.  (For instance, the intersection graph of balls describes
connectivity in a network of sensors with bounded transmission radius.)

Previously, nontrivial fully dynamic results were known only
for special cases like axis-parallel line segments and rectangles.
We provide similarly improved
update times, $\Ot(n^{2/3})$, for these special cases.
Moreover, we show how to obtain sublinear update bounds for
virtually {\em all\/} families 
of geometric objects which allow sublinear-time range queries. 
In particular, we obtain the {\em first\/} sublinear update time
for arbitrary 2D line segments: $\Ost(n^{9/10})$; for $d$-dimensional
simplices: $\Ost(n^{1-\frac{1}{d(2d+1)}})$; and
for $d$-dimensional balls: $\Ost(n^{1-\frac{1}{(d+1)(2d+3)}})$.

%It was known that for objects which allow polylogarithmic range
%queries (most notably, axis-parallel rectangles), geometric
%connectivity reduces to subgraph connectivity. With our improved
%subgraph connectivity techniques, we show that geometric connectivity
%is possible in sublinear time, for any objects which allow
%sublinear-time range queries. This is a very broad class of objects,
%including all semialgebraic sets of constant description complexity.
\end{abstract}

\thispagestyle{empty}
\newpage

\section{Introduction}

\subsection{Dynamic Graphs}

%\paragraph{The model.}
Dynamic graphs inspire a natural, challenging, and well-studied class
of algorithmic problems. A rich body of the STOC/FOCS literature has
considered problems ranging from the basic question of understanding
connectivity in a dynamic graph \cite{frederickson85connect,
henzinger99connect, thorup00connect, chan02subgraph,
patrascu07edgedel}, to maintaining the minimum spanning tree
\cite{holm01connect}, the min-cut \cite{thorup07mincut}, shortest
paths \cite{demetrescu03apsp, thorup05apsp}, reachability in directed
graphs \cite{demetrescu05reach, 
king99reach, king02reach, roditty04reach, sankowski04reach}, etc.

But what exactly makes a graph ``dynamic''? Computer networks have
long provided the common motivation. % for studying dynamic graphs. 
The
dynamic nature of such networks is captured by two basic types of
updates to the graph:
% that can be applied to the graph:
%
%\vspace{-0.5ex}
\begin{myitemize}
\item edge updates: adding or removing an edge. These correspond to
  setting up a new cable connection, accidental cable cuts, etc.

\item vertex updates: turning a vertex on and off. Vertices (routers)
  can temporarily become ``off'' after events such as a
  misconfiguration, a software crash and reboot, etc.  Problems
  involving only vertex updates have been called {\em dynamic
  subgraph} problems, since queries refer to the subgraph
  induced by vertices which are on.
\end{myitemize}

Loosely speaking, dynamic graph problems fall into two categories.
For ``hard'' problems, such as shortest paths and directed
reachability, the best known running times are at least linear in the
number of vertices. These high running times obscure the difference
between vertex and edge updates, and identical bounds are often
stated \cite{demetrescu03apsp, roditty04reach, sankowski04reach} for
both operations. For the remainder of the problems, sublinear
running times are known for edge updates, but sublinear
bounds for vertex updates seems much harder to get.   For instance,
even iterating
through all edges incident to a vertex may take linear time in the worst case.
That vertex
updates are slow is unfortunate.  Referring to the computer-network metaphor, 
vertex updates are cheap
``soft'' events (misconfiguration or reboot), which occur more
frequently than the costly physical events (cable cut) that cause an
edge update.

\paragraph{Subgraph connectivity.}
As mentioned, most previous sublinear dynamic graph algorithms address edge
updates but not the equally fundamental vertex updates.
One notable exception, however,
was a result of Chan \cite{chan02subgraph} from
STOC'02 on the basic connectivity
problem for general sparse (undirected) graphs. This algorithm
can support vertex updates in time%
\footnote{We use $m$ and $n$ to denote the number of edges and vertices of 
the graph respectively;
$\Ot(\cdot)$ ignores polylogarithmic factors and $\Ost(\cdot)$ hides
$n^\eps$ factors for an arbitrarily small constant $\eps>0$. 
%  $\Ot(T) = T \cdot \lg^{O(1)} m$ ignores polylogarithmic factors. 
  Update bounds in this paper are, by default, amortized.}
$O(m^{0.94})$ and decide whether two query vertices are connected
in time
$\Ot(m^{1/3})$.

Though an encouraging start, the nature of this result makes it appear more
like a half breakthrough. For one, the update time is only slightly
sublinear. Worse yet, Chan's algorithm requires fast matrix
multiplication (FMM\@). The
$O(m^{0.94})$ update time follows from the theoretical FMM algorithm
of Coppersmith and Winograd \cite{coppersmith89fmm}. If Strassen's
algorithm is used instead, the update time becomes $O(m^{0.984})$. 
Even if optimistically FMM could be done in quadratic time, the
update time would only improve to $O(m^{0.89})$.
FMM has been used before in various dynamic graph algorithms (e.g.,
\cite{demetrescu05reach, king02reach}), and
the paper \cite{chan02subgraph} noted specific connections 
to some matrix-multiplication-related problems (see Section~\ref{sec:related}).
All this naturally led one to suspect, as conjectured in the paper, that 
FMM might be essential to our problem.  
Thus, the result we are about to describe may come as
a bit of a surprise\ldots
%the previous paper
%\cite{chan02subgraph} declared ``{\em [w]e suspect
%that [\ldots] FMM is essential to solve our problem}.''
%Fortunately (and rather surprisingly),
%we will show that this pessimism was unfounded.

% m^{4\omega / (3\omega + 3)}
%Here $\omega$ is the exponent in fast matrix multiplication (FMM).

%Motivated by these uninspiring aspects, Chan asked whether sublinear
%update time is contingent on FMM, and declared ``{\em [w]e suspect

\paragraph{Our result.}
In this paper, we present a new algorithm for 
dynamic connectivity,
achieving an improved vertex-update time of $\Ot(m^{2/3})$,
with an identical query time of $\Ot(m^{1/3})$.  First of all,
this is a significant
\emph{quantitative} improvement %over the state of the art 
(to anyone
who regards an $m^{0.27}$ factor as substantial), and it represents
the first convincingly sublinear running time.  
More importantly,
it is a significant \emph{qualitative} improvement, as our
bound does not require FMM\@. 
Our algorithm involves a number of ideas, some of which can be traced back
to earlier algorithms, but we use known edge-updatable connectivity 
structures to maintain a more cleverly designed intermediate 
graph.  The end product is not straightforward at all, but still 
turns out to be \emph{simpler} than the previous method~\cite{chan02subgraph} 
and has a compact, two-page description (we regard this as another plus,
not a drawback). 

%Remarkably, our algorithm is also
%\emph{simpler}, and contains (in our opinion) a beautiful and
%nontrivial usage of the ``high-degree vs.\ low-degree'' idea
%seen in various static graph algorithms \cite{alon97triangles,
%yuster04triangles}.

\subsection{Dynamic Geometry}

We next turn to another important class of dynamic connectivity problems---those
arising from geometry.

\paragraph{Geometric connectivity.}

%[TC] delete for now, to keep things concise
%
%A mantra of theoretical computer science is that the study of
%fundamental problems should not be tied to an immediate
%application: if the problem is natural, it will find new, surprising
%applications. In the case of dynamic subgraph connectivity, one did
%not have to wait long for this promise to materialize---the problem
%provides a foothold in a seemingly distant realm of {\em geometric
%connectivity}.
%%, and implies the first nontrivial result for most
%%problems in this class.

Consider the following question, illustrated in
Figure~\ref{fig:drawing}(a).  Maintain a set of line segments in the
plane, under insertions and deletions, to answer queries of the form:
``given two points $a$ and $b$, is there a path between $a$ and $b$
along the segments?''

%\IGNORE{
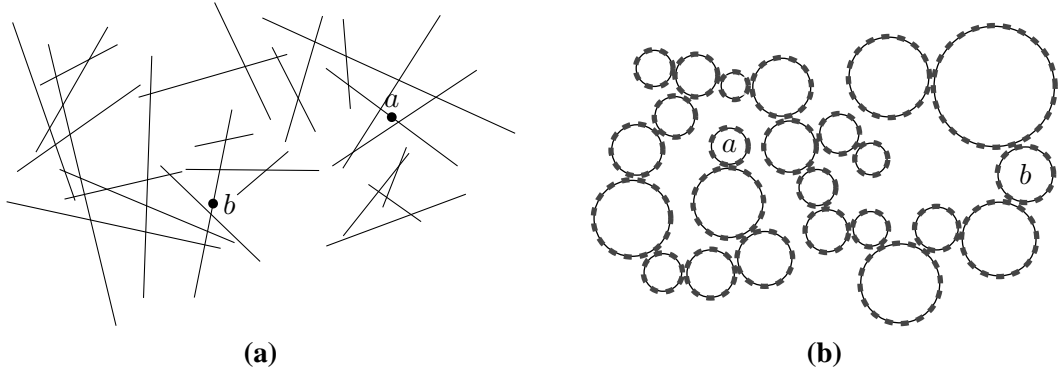
\begin{figure*}
  \centering
  \begin{tabular}{c@{\hspace{1cm}}c}
    \begin{tikzpicture}[scale=0.05]
\draw (-8.3, 58.2)  -- (8.3, 10.8);
\draw (37.7, 19.1) -- (73.2, 18.8);
\draw (25.2, 38.3) -- (64.7, 49.7);
\draw (-1, 41.5)  -- (19.6, 52.3);
\draw (51.4, 12.3)  -- (65.0, 23.9);
\draw (17.0, 57.0)  -- (-2.1, 23.7);
\draw (79.6, 59.1)  -- (81.5, 35.2);
\draw (47.0, -1.9)  -- (-9.9, 10.5);
\draw (31.0, 20.1)  -- (57.5, -5.4);
\draw (97.1, 23.4)  -- (79.6, 1.4);
\draw (40, -15) -- (50, 35);
  \fill[black] (45, 10) node[right] {$b$} circle (1.3);
\draw (1.2, 52.4)  -- (19.2, -22.5);
\draw (90.1, 9.0)  -- (96.4, 25.0);
\draw (4.2, 19.1)  -- (50.7, -0.4);
\draw (40.1, 25.2) -- (55.7, 28.5);
\draw (28.7, 49.3) -- (26.5, -15.1);
\draw (5.5, 11.2)  -- (36.8, 18.5);
\draw (110, 20)  -- (75, 46);
  \fill[black] (92.5, 33) node[above] {$a$} circle (1.3);
\draw (72.2, 29.1)  -- (60.7, 51.6);
\draw (115.1, 45.5) -- (76.7, 20.0);
\draw (-7.1, 18.5)  -- (25.7, 41.6);
\draw (112.2, 12.5) -- (75.1, -1.3);
\draw (58.1, 59.4)  -- (125.3, 28.7);
\draw (60.3, 32.3)  -- (45.4, 63.5);
\draw (105.4, 60.1) -- (79.5, 19.4);
\draw (64.2, 26.4)  -- (74.1, 59.9);
\draw (86.3, 15.1)  -- (100.3, 5.2);
    \end{tikzpicture}
    &
    \begin{tikzpicture}[scale=0.05]
    \draw[line width=0.5pt] (-0.8, 30.2) node[white] {0} circle (6.7);
      \draw[dashed, black!73, line width=2pt] (-0.8, 30.2) circle (6.9);
    \draw[line width=0.5pt] (61.5, 27.8) node[white] {1} circle (4.2);
      \draw[dashed, black!73, line width=2pt] (61.5, 27.8) circle (4.4);
    \draw[line width=0.5pt] (6.1, -2.4) node[white] {2} circle (5.0);
      \draw[dashed, black!73, line width=2pt] (6.1, -2.4) circle (5.2);
    \draw[line width=0.5pt] (9.3, 39.1) node[white] {5} circle (5.3);
      \draw[dashed, black!73, line width=2pt] (9.3, 39.1) circle (5.5);
    \draw[line width=0.5pt] (15.0, 50.0) node[white] {7} circle (5.5);
      \draw[dashed, black!73, line width=2pt] (15.0, 50.0) circle (5.7);
    \draw[line width=0.5pt] (69.2, -5.3) node[white] {10} circle (10.5);
      \draw[dashed, black!73, line width=2pt] (69.2, -5.3) circle (10.7);
    \draw[line width=0.5pt] (-2.1, 12.0) node[white] {14} circle (10.1);
      \draw[dashed, black!73, line width=2pt] (-2.1, 12.0) circle (10.3);
    \draw[line width=0.5pt] (102.5, 23.7) node {$b$} circle (7.3);
      \draw[dashed, black!73, line width=2pt] (102.5, 23.7) circle (7.5);
    \draw[line width=0.5pt] (94.2, 47.1) node[white] {17} circle (16.0);
      \draw[dashed, black!73, line width=2pt] (94.2, 47.1) circle (16.2);
    \draw[line width=0.5pt] (23.8, 31.2) node {$a$} circle (4.8);
      \draw[dashed, black!73, line width=2pt] (23.8, 31.2) circle (5.0);
    \draw[line width=0.5pt] (33.2, 1.3) node[white] {20} circle (7.2);
      \draw[dashed, black!73, line width=2pt] (33.2, 1.3) circle (7.4);
    \draw[line width=0.5pt] (47.1, 20.2) node[white] {21} circle (4.8);
      \draw[dashed, black!73, line width=2pt] (47.1, 20.2) circle (5.0);
    \draw[line width=0.5pt] (95.5, 6.5) node[white] {23} circle (9.8);
      \draw[dashed, black!73, line width=2pt] (95.5, 6.5) circle (10.0);
    \draw[line width=0.5pt] (3.8, 51.9) node[white] {24} circle (4.8);
      \draw[dashed, black!73, line width=2pt] (3.8, 51.9) circle (5.0);
    \draw[line width=0.5pt] (23.5, 16.1) node[white] {25} circle (9.2);
      \draw[dashed, black!73, line width=2pt] (23.5, 16.1) circle (9.4);
    \draw[line width=0.5pt] (25.1, 47.3) node[white] {26} circle (3.5);
      \draw[dashed, black!73, line width=2pt] (25.1, 47.3) circle (3.7);
    \draw[line width=0.5pt] (53.0, 34.5) node[white] {27} circle (5.2);
      \draw[dashed, black!73, line width=2pt] (53.0, 34.5) circle (5.4);
    \draw[line width=0.5pt] (49.7, 8.7) node[white] {28} circle (5.6);
      \draw[dashed, black!73, line width=2pt] (49.7, 8.7) circle (5.8);
    \draw[line width=0.5pt] (37.7, 46.9) node[white] {32} circle (7.7);
      \draw[dashed, black!73, line width=2pt] (37.7, 46.9) circle (7.9);
    \draw[line width=0.5pt] (61.3, 9.2) node[white] {33} circle (4.6);
      \draw[dashed, black!73, line width=2pt] (61.3, 9.2) circle (4.8);
    \draw[line width=0.5pt] (18.8, -2.7) node[white] {35} circle (6.3);
      \draw[dashed, black!73, line width=2pt] (18.8, -2.7) circle (6.5);
    \draw[line width=0.5pt] (78.9, 9.3) node[white] {36} circle (5.7);
      \draw[dashed, black!73, line width=2pt] (78.9, 9.3) circle (5.9);
    \draw[line width=0.5pt] (66.2, 49.6) node[white] {38} circle (10.6);
      \draw[dashed, black!73, line width=2pt] (66.2, 49.6) circle (10.8);
    \draw[line width=0.5pt] (39.9, 31.1) node[white] {39} circle (6.8);
      \draw[dashed, black!73, line width=2pt] (39.9, 31.1) circle (7.0);
    \end{tikzpicture}
    \\
    {\bf (a)} & {\bf (b)} \\
  \end{tabular}

  \caption{ (a) Is $b$ reachable from $a$ staying on the roads? ~~
    (b) Do the gears transmit rotation from $a$ to $b$? }
  \label{fig:drawing}
\end{figure*}
%}

%\begin{figure}
%  \centering
%    \begin{tikzpicture}[scale=0.045]
%      \input{lines.tik}
%    \end{tikzpicture}
%  \caption{ Is $b$ reachable from $a$ staying on the roads?}
%  \label{fig:drawing}
%\vspace{-3ex}
%\end{figure}

This simple-sounding problem turns out to be a challenge.
%, and after a moment's thought, obtaining a sublinear update
%time may appear quite hopeless. 
%The obstacles in front of a fast
%dynamic solution lie in the dual nature of the problem. 
On one
hand, understanding any local geometry does not seem to help, because
the connecting path can be long and windy. On the other hand, the
graph-theoretic understanding is based on the intersection graph,
which is too expensive to maintain. A newly inserted (or deleted)
segment can intersect a large number of objects in the set, changing
the intersection graph dramatically.

Abstracting away, we can consider a broad class of problems of the
form: maintain a set of $n$ geometric objects, and answer connectivity
queries in their intersection graph. Such graphs arise,
for instance, in VLSI applications in the case of orthogonal
segments, or gear transmission systems, in the
case of touching disks; see Figure~\ref{fig:drawing}(b).
%\ref{fig:drawingb}%
A more
compelling application can be found in sensor networks: if $r$ is the
radius within which two sensors can communicate, the communication
network is the intersection graph of balls of radius $r/2$ centered
at the sensors. 
While our focus is on theoretical understanding rather than
the practicality of 
specific applications, these examples still indicate the
natural appeal of geometric connectivity problems.
%It should be noted that we do not claim that any
%degree of practicality is associated with these
%applications. Nonetheless, we find that these natural examples make
%geometric connectivity an appealing problem.

%\begin{figure}[h]
%  \centering
%    \begin{tikzpicture}[scale=0.038]
%      \input{gears.tik}
%    \end{tikzpicture}
%  \caption{ Do the gears transmit rotation from $a$ to $b$?}
%  \label{fig:drawingb}
%\vspace{-3ex}
%\end{figure}

All these problems have a trivial $\Ot(n)$ solution, by maintaining
the intersection graph through edge updates. A systematic approach to
beating the linear time bound was proposed in
Chan's paper as well~\cite{chan02subgraph}, by drawing a %remarkable
connection to
subgraph connectivity. Assume that a particular object type allows
data structures for intersection range searching with space $S(n)$ and
query time $T(n)$. It was shown that geometric
connectivity can essentially be solved by maintaining a graph of size $m =
O(S(n)+nT(n))$ and running $O(S(n)/n+T(n))$ vertex updates for every object insertion
or deletion.  Using the previous subgraph connectivity result
\cite{chan02subgraph}, an update in the geometric connectivity problem
took time $\Ot([S(n)/n + T(n)] \cdot [S(n)+nT(n)]^{0.94})$. Using our improved result,
the bound becomes $\Ot([S(n)/n + T(n)] \cdot [S(n)+nT(n)]^{2/3})$.

The prime implication in the previous paper is that connectivity of
axis-parallel boxes in any constant dimension (in particular,
orthogonal line segments in the plane) reduces to subgraph
connectivity, with a polylogarithmic cost. Indeed, for such
boxes range trees yield $S(n) = n
\cdot \lg^{O(d)} n$ and $T(n) = \lg^{O(d)} n$.
Unfortunately, while nontrivial range searching results are
known for many types of %exotic 
objects, very efficient range searching
is hard to come by. Consider our main motivating examples: 
%
%\vspace{-0.5ex}
\begin{myitemize}
\item for arbitrary (non-orthogonal) line segments in $\R^2$, one can achieve 
%\footnote{ There are known trade-off results which obtain better query
%  time with more than linear space. However, these results do not seem
%  to help in our application, because the space degrades too quickly
%  compared to the improvement in query time. }
%
  $T(n) = \Ost(\sqrt{n})$ and $S(n) = \Ost(n)$, or
  $T(n) = \Ost(n^{1/3})$ and $S(n)=\Ost(n^{4/3})$ \cite{matousek92partition}.

\item for disks in $\R^2$, one can achieve $T(n) = \Ost(n^{2/3})$ and $S(n) =
  \Ost(n)$, or 
  $T(n) = \Ost(n^{1/2})$ and $S(n)=\Ost(n^{3/2})$ \cite{agarwal94range}.
\end{myitemize}
Even with our improved vertex-update time, 
the $[S(n)/n + T(n)] \cdot [S(n)+nT(n)]^{2/3}$
bound is too weak to beat the trivial linear
update time.  %To get a sublinear update time f
For arbitrary line
segments in $\R^2$, one would need to improve the vertex-update time
to $m^{1/2 - \eps}$, which appears
unlikely without FMM (see Section~\ref{sec:related}).  
The line segment case was in fact
mentioned as a major open problem, implicitly in \cite{chan02subgraph}
and explicitly in \cite{afshani06connectivity}.
The situation gets worse
for objects of higher complexity or in higher dimensions.

%imply an improvement to a long-standing bound
%on finding triangles in graphs.

%Thus, handling non-orthogonal line segments or disks seems to hit a
%fundamental barrier. Handling non-orthogonal line segments in
%particular was a central open problem of \cite{chan02subgraph},
%reiterated as the main open problem of \cite{afshani06connectivity}.

\paragraph{Our results.}
%We describe further new ideas, on top of our improved subgraph
%connectivity algorithm, which 
In this paper, we are finally able to break the above barrier
for dynamic geometric connectivity.
%Building on our new ideas for vertex updates, we describe a technique
%for handling geometric connectivity, which is effective for a much
%richer class of objects. 
At a high level, we show that range searching
with \emph{any} sublinear query time is enough to obtain sublinear
update time in geometric connectivity. In particular, we get the
\emph{first} nontrivial update times for arbitrary line segments in the
plane, disks of arbitrary radii, and
simplices and balls in any fixed dimension.
While the previous reduction \cite{chan02subgraph} involves merely a
straightforward usage of ``biclique covers'', 
our result here requires much more work.  For starters, we need to devise
a ``degree-sensitive'' version of our improved subgraph connectivity algorithm 
(which is of interest in itself); we then use this and known connectivity
structures to maintain not one but two carefully designed intermediate
graphs.

Essentially, if $T(n) = \Ot(n^{1-b})$ and $S(n) = \Ot(n)$, we can support
dynamic geometric connectivity with update time $\Ot \big( n^{1 - b^2
/ (2+b)} \big)$ and query time $\Ot\big( n^{b / (2+b)} \big)$. For
non-orthogonal line segments in $\R^2$, this gives an update time of
$\Ost(n^{9/10})$ and a query time of $\Ost(n^{1/5})$.
%, answering the open
%question from \cite{chan02subgraph, afshani06connectivity}. 
For disks in $\R^2$,
the update time is $\Ost(n^{20/21})$, with a query time of
$\Ost(n^{1/7})$.

%Our results indeed deliver the promise of sublinear update time to a
%truly broad class of geometric shapes, since 
Known range searching
techniques \cite{agarwal-erickson} from computational geometry almost
always provide sublinear query time. For instance,
Matou\v{s}ek~\cite{matousek92partition} showed that $b\approx 1/2$ is
attainable for line segments, triangles, and any constant-size
polygons in $\R^2$; more generally, $b\approx 1/d$ for
simplices or constant-size polyhedra in $\R^d$.  Further results by
Agarwal and Matou\v{s}ek~\cite{agarwal94range} yield $b\approx 1/(d+1)$ for
balls in $\R^d$.  Most generally, $b>0$ is possible for
\emph{any} class of objects defined by semialgebraic sets 
of constant description complexity.

\paragraph{More results.}
%While our main contribution on dynamic geometric connectivity
%is in providing the first proof of the existence of nontrivial
%bounds in very general settings, 
Our general sublinear results undoubtedly
invite further research into finding better
bounds for specific classes of objects.
In general, the complexity of range queries provides a natural
barrier for the update time, since upon inserting an object we at
least need to determine if it intersects any object already in the
set. Essentially, our result has a quadratic loss compared to
range queries: if $T(n) = n^{1-b}$, the update time is
$n^{1-\Theta(b^2)}$.
%It would be interesting to see if this quadratic gap can be reduced.
%in a black-box
%fashion.  

In Section~\ref{app:offline}, 
We make a positive step towards
closing this quadratic gap: 
we show that if the updates are given \emph{offline}
(i.e.~are known in advance), the amortized update time can be
made $n^{1-\Theta(b)}$. 
We need FMM this time, but the usage of FMM here is more intricate
(and interesting) than typical.  For one, it is crucial to use
fast {\em rectangular\/} matrix multiplication.  Along the way,
we even find ourselves rederiving Yuster and Zwick's sparse matrix 
multiplication result~\cite{yuster04FMM} in a more general form.  
The juggling of parameters is also more unusual, as one can suspect from
looking at our actual update bound, which is
$\Ot(n^{\frac{1+\alpha-b\alpha}{1+\alpha-b\alpha/2}})$, where
$\alpha=0.294$ is an exponent associated with rectangular FMM\@.
%Due to space limitation, the description of our improved
%offline algorithms is deferred to the full paper.

%Unfortunately, this result relies on FMM,
%though {\em rectangular\/} matrix multiplication is used in
%an interesting way.

\section{Related Work}\label{sec:related}

Before proceeding to our new algorithms, we mention more
related work, for the sake of completeness.

\paragraph{Graphs.}
Most previous work on dynamic subgraph connectivity concerns special cases only.
%Naturally, faster solutions are possible in special cases.  
Frigioni
and Italiano~\cite{frigioni00subgraph} considered vertex updates in
planar graphs, and described a polylogarithmic solution.

If vertices have constant degree, vertex updates are equivalent to
edge updates. For edge updates, Henzinger and
King~\cite{henzinger99connect} were first to obtain polylogarithmic
update times (randomized). This was improved by Holm et
al.~\cite{holm01connect} to a deterministic solution with $O(\lg^2 m)$
time per update, and by Thorup~\cite{thorup00connect} to a randomized
solution with $O(\lg m \cdot (\lg\lg m)^3)$ update time. The
randomized bound almost matches the $\Omega(\lg m)$ lower bound from
\cite{patrascu06loglb}.
All these data structures maintain a spanning forest as a certificate
for connectivity.
%, making the query time slightly sublogarithmic. 
This
idea fails for vertex updates in the general case, since the 
certificate can change
substantially after just one update.

% [TC] don't hate me for deleting this paragraph! (it sounds
% too defensive and draws a tad too much attention to applications)
% replace with one (bland) sentence
%
\IGNORE{
On motivational ground, one may ask whether supporting constant-degree
graphs is not enough for computer networks. It is indeed true that
physical, administrative, and cost constraints make the geographic
layout of a network look like a planar graph of small degree. However,
the large latency impact of going through the electric domain in a
router, coupled with low-cost fiber optics, have pushed for designs in
which the physical layer of the network is a more complex overlay
graph over the geographic graph. For example, a busy router in New
York might have direct optical connections (edges) both to a router in
Philadelphia, and one in Atlanta, even though the optical cable to
Atlanta passes physically through the New York--Philadelphia conduit.
}%%%%%%%
In many practical settings, these planar-graph and constant-degree
special cases are unfortunately inadequate. In particular, large
networks of routers are often designed as overlay graphs over a
(small-degree) geographic graph. Long fiber-optic links bypass
intermediate nodes, in order to minimize the latency cost of passing
through the electric domain repeatedly.

For more difficult dynamic graph problems, the goal is typically changed
from getting polylogarithmic bounds to finding better exponents in
polynomial bounds; for example,
see all the papers on directed reachability
\cite{demetrescu05reach, king99reach, roditty04reach, sankowski04reach}.
Evidence suggests that dynamic subgraph connectivity fits this category.
%Regarding hardness of the general dynamic subgraph connectivity problem, 
It was observed~\cite{chan02subgraph} that 
finding
triangles (3-cycles) or quadrilaterals (4-cycles) in directed graphs
can be reduced to $O(m)$ vertex updates.  Thus, an update bound
better than $\sqrt{m}$ appears unlikely without FMM,
since the best running time for
finding triangles without FMM is $O(m^{3/2})$, dating back to
STOC'77 \cite{itai78triangles}. Even with FMM, known results are only
slightly better: finding triangles and quadrilaterals takes time
$O(m^{1.41})$ \cite{alon97triangles} and $O(m^{1.48})$ 
\cite{yuster04triangles} respectively.  Thus, current knowledge
prevents an update bound better than $m^{0.48}$.
%This suggests that 
%the best vertex-update 
%bound we can hope for is $m^{\Omega(1)}$.

\paragraph{Geometry.}
%Dynamic geometric connectivity is limited not only by the complexity
%of range searching but also by
%%When very fast range searching is possible (axis-parallel boxes),
%%efficiency is not limited by the range query, but by reliance on
%dynamic subgraph connectivity. %Interestingly,
%Chan~\cite{chan02subgraph} 
It was shown~\cite{chan02subgraph} that subgraph connectivity can be
reduced to dynamic connectivity of axis-parallel line segments in 3
dimensions. Thus, as soon as one gets enough combinatorial richness in
the host geometric space, subgraph connectivity becomes the
\emph{only} possible way to solve geometric connectivity.

When the geometry is less combinatorially rich, it is possible to find
\emph{ad hoc} algorithms that do not rely on subgraph
connectivity. Special cases that have been investigated include the 
following:

\begin{myitemize}
\item for orthogonal segments or axis-parallel
  rectangles in the plane, Afshani and
  Chan~\cite{afshani06connectivity} proposed a data structure with
  update time $\Ot(n^{10/11})$ and constant query time. This is
  incomparable to our result of update time $\Ot(n^{2/3})$ and query
  time $\Ot(n^{1/3})$.

\item for unit axis-parallel hypercubes, the problem 
%is essentially equivalent  [TC:why?]
  reduces to maintaining the minimum spanning tree under the
  $\ell_\infty$ metric. Eppstein~\cite{eppstein95emst} describes a
  general technique for dynamic geometric MST, ultimately appealing to
  range searching, and obtains polylogarithmic time per operation.

\item for unit balls, the problem reduces to dynamic Euclidean
  MST, which in turn reduces to range searching by Eppstein's
  technique~\cite{eppstein95emst}. In two dimensions, Chan's dynamic
  nearest-neighbor data structure~\cite{chan06dynNN} implies an
  $O(\lg^{10} n)$ update time for this problem.

%\item Agarwal and van Kreveld~\cite{agarwal96connected} investigated
%  incremental (insertion-only) connectivity problems. For example, in
%  the case of line segments in the plane, they obtain
%  $O(n^{1/3+\eps})$ amortized time per insertion, for any fixed
%  $\eps>0$.
\end{myitemize}

Dynamic geometric connectivity is a natural continuation of static
geometric connectivity problems, which have been studied since the
early 1980s. As in our case, the main challenge is to avoid working
explicitly with the intersection graph, which could be of quadratic
size. Known results include $O(n\lg n)$-time
algorithms~\cite{imai82intersect, imai83intersect} for computing the
connected components of axis-aligned rectangles in the plane, and
$\Ot(n^{4/3})$-time algorithms~\cite{guibas88intersect,
lopez95intersect} for arbitrary line segments in the plane. More
generally, Chan~\cite{chan02subgraph} (and later 
Eppstein~\cite{eppstein04bipart}) noted the connection of static
geometric connectivity to range searching, %via biclique covers, 
which
implied subquadratic algorithms for objects with constant description
complexity.  The connection carries over to 
the incremental (insertion-only) and decremental (deletion-only)
cases~\cite{chan02subgraph}, e.g., yielding $\Ot(n^{1/3})$ update time 
for arbitrary line segments, reproving and extending some older 
results~\cite{agarwal96connected}.

Another related problem is maintaining %an understanding of 
connectivity in
the kinetic setting, where objects move continuously according to known
flight plans. See \cite{hershberger99kinetic,
hershberger01kinetic} for the case of axis-parallel boxes, and
\cite{guibas01kinetic} for unit disks.

\section{Dynamic Subgraph Connectivity with $\Ot(m^{2/3})$ Update Time}

In this section, we present our new method for the
dynamic subgraph connectivity problem: maintaining a subset $S$
of vertices in a graph $G$, under vertex insertions and deletions in $S$, 
so that we can decide whether any two query vertices are connected in 
the subgraph induced by $S$.  We will call the vertices in $S$ the 
{\em active\/} vertices.  For now, we assume that the graph $G$ itself
is static.

The complete description of the new method is given in the proof 
of the following theorem.
It is ``short and sweet'', especially if the reader compares with Chan's
paper~\cite{chan02subgraph}.  The previous method requires
several stages of development, addressing
the offline and semi-online special cases, along with the use of FMM---we
completely bypass these intermediate stages, and FMM, here.
Embedded below, one can find a number
of different ideas (some also used in~\cite{chan02subgraph}):
rebuilding periodically after a certain number of updates,
distinguishing ``high-degree'' features
from ``low-degree'' features (e.g., see \cite{alon97triangles,
yuster04triangles}), amortizing by splitting smaller subsets from
larger ones, etc.  The key lies in the definition of a new, yet
deceptively simple,
intermediate graph $G^*$, which is maintained by known
polylogarithmic data structures for dynamic connectivity under
edge updates \cite{henzinger99connect,
holm01connect, thorup00connect}.
Except for these known connectivity structures, the description is 
entirely self-contained.

%the challenge lies in fitting them together in the
%right way.\footnote{For readers familiar with \cite{chan02subgraph}, 
%one key difference lies in the way we let polylogarithmic connectivity 
%data structures maintain a larger graph $G^*$.}
%The description below can be enjoyed without reading previous
%papers, and is essentially self-contained, except for the
%use of some known polylogarithmic data structures for dynamic
%connectivity under edge updates .

\begin{theorem}\label{conn}
We can design a data structure for dynamic subgraph connectivity
for a graph $G=(V,E)$ with $m$ edges, having amortized vertex update time
$\Ot(m^{2/3})$, query time $\Ot(m^{1/3})$, and preprocessing time
$\Ot(m^{4/3})$.
\end{theorem}
\begin{proof}
We divide the update sequence into phases, each consisting of $q :=
m/\Delta$ updates.  The active vertices are partitioned into two sets $P$
and $Q$, where $P$ undergoes only deletions and $Q$ undergoes both
insertions and deletions.  Each vertex insertion is done to $Q$.  At
the end of each phase, we move the elements of $Q$ to $P$ and reset
$Q$ to the empty set.  This way, $|Q|$ is kept at most $q$ at all
times.

Call a
connected component in (the subgraph induced by) $P$ {\em high\/} if
the sum of the degrees of its vertices exceeds $\Delta$, and {\em low\/}
otherwise.  Clearly, there are at most $O(m/\Delta)$ high components.

\paragraph{The data structure.}
\begin{myitemize}
\item We store the components of $P$ in a 
data structure for decremental (deletion-only) connectivity 
that supports edge deletions in polylogarithmic amortized time.
\item We maintain a bipartite multigraph $\Gamma$ between $V$
and the components $\gamma$ in $P$: for each $uv\in E$
where $v$ lies in component $\gamma$, we create a copy of
an edge $u\gamma\in\Gamma$.
\item
For each vertex pair $u$,$v$,
we maintain the value $C[u,v]$ defined as
the number of low components in $P$ 
that are adjacent to both $u$ and $v$ in $\Gamma$.  (Actually,
only $O(m\Delta)$ entries of $C[\cdot,\cdot]$ are nonzero 
and need to be stored.)
\item
We define a graph $G^*$ whose vertices are the vertices of $Q$
and components of $P$:
\begin{enumerate}
\item[(a)] For each $u,v\in Q$,
if $C[u,v]>0$, then create an edge $uv\in G^*$.
\item[(b)] For each vertex $u\in Q$ and high component $\gamma$ in $P$,
if $u\gamma\in\Gamma$, then create an edge $u\gamma\in G^*$.
\item[(c)] For each $u,v\in Q$,
if $uv\in E$, then create an edge $uv\in G^*$.
\end{enumerate}
We maintain $G^*$ in another data structure for 
dynamic connectivity supporting polylogarithmic-time edge updates.
\end{myitemize}

\paragraph{Justification.}
We claim that two vertices of $Q$ are connected in the subgraph induced
by the active vertices in $G$ iff they
are connected in $G^*$.  The ``if'' direction is obvious.
For the ``only if'' direction, suppose two vertices $u,v\in Q$ are 
``directly'' connected in $G$ by being adjacent to a common 
component $\gamma$ in $P$.  If $\gamma$ is high, then edges
of type~(b) ensure that $u$ and $v$ are connected in $G^*$.  If
instead $\gamma$ is low, then edges of type (a) ensure that $u$ and $v$
are connected in $G^*$.  
By concatenation, the argument extends to show that any two
vertices $u,v\in Q$ connected by a path in $G$ are connected in $G^*$.

\paragraph{Queries.}  
Given two vertices $v_1$ and $v_2$, if both are in $Q$, 
we can simply test whether they are connected in~$G^*$.

If instead $v_j\ (j\in\{1,2\})$ is in a high component $\gamma_j$, then we
can replace $v_j$ with any vertex of $Q$ adjacent to $\gamma_j$ in $G^*$.
If no such vertex exists, then because of type-(b) edges,
$\gamma_j$ is an isolated component
and we can simply test whether $v_1$ and $v_2$ are both in the same 
component of $P$.

If on the other hand $v_j$ is in a low component $\gamma_j$,
then we can exhaustively search for a vertex in $Q$ adjacent to 
$\gamma_j$ in $\Gamma$, in $\Ot(\Delta)$ time,
and replace $v_j$ with such a vertex.  
Again if no such vertex exists, then $\gamma_j$ is an isolated component
and the test is easy.  The query cost is $\Ot(\Delta)$.

\paragraph{Preprocessing per phase.}  
At the beginning of each phase,
we can compute the multigraph $\Gamma$ in $\Ot(m)$ time.  We can
compute the matrix $C[\cdot,\cdot]$ in $\Ot(m\Delta)$ time, by examining
each edge $ v\gamma \in\Gamma$ and each of the $O(\Delta)$ vertices $u$
adjacent to a low component $\gamma$
and testing whether $\gamma u\in\Gamma$.  The graph $G^*$ can then
be initialized.  The cost per phase is $\Ot(m\Delta)$.
We can cover this cost by charging every update operation
with amortized cost $\Ot(m\Delta/q) = \Ot(\Delta^2)$.

\paragraph{Update of a vertex $u$ in $Q$.}  
We need to update $O(q)$ edges of
types~(a) and (c), and $O(m/\Delta)$ edges of type~(b) in $G^*$.
The cost is $\Ot(q + m/\Delta) = \Ot(m/\Delta)$.

\paragraph{Deletion of a vertex from a low component $\gamma$ in $P$.}
The component $\gamma$ is split into a number of subcomponents.
Since the total degree in $\gamma$ is $O(\Delta)$, we can update
the multigraph $\Gamma$ in $\Ot(\Delta)$ time.
Furthermore, we can update the matrix $C[\cdot,\cdot]$ in $\Ot(\Delta^2)$
time, by examining each vertex pair $u,v$ adjacent to $\gamma$ and
decrementing $C[u,v]$ if $u$ and $v$ lie in different subcomponents.
Consequently, we need to update $O(\Delta^2)$ edges of type~(a).  
The cost is $\Ot(\Delta^2)$.

\paragraph{Deletion of a vertex from a high component $\gamma$ in $P$.}
The component $\gamma$ is split into a number of subcomponents
$\gamma_1\TO\gamma_\ell$ with, say, $\gamma_1$ being the largest.
We can update the multigraph $\Gamma$ in time $\Ot(\deg(\gamma_2)+
\cdots+\deg(\gamma_\ell))$ by splitting the smaller subcomponents from
the largest subcomponent.  Consequently, we need to update
$O(\deg(\gamma_2)+\cdots+\deg(\gamma_\ell))$ edges
of type (b) in $G^*$.  Since $P$ undergoes deletions only,
a vertex can belong to the
smaller subcomponents in at most $O(\lg n)$ splits over the entire
phase, and so the total cost per phase is $\Ot(m)$, which
is absorbed in the preprocessing cost of the phase.

For each low subcomponent $\gamma_j$, we update the matrix 
$C[\cdot,\cdot]$ in $\Ot(\deg(\gamma_j)\Delta)$ time, by examining
each edge $\gamma_j v\in\Gamma$ and each of the $O(\Delta)$ vertices $u$
adjacent to $\gamma_j$
and testing whether $\gamma_j u\in\Gamma$.  Consequently, we need
to update $O(\deg(\gamma_j)\Delta)$ edges of type~(a) in $G^*$.  Since
a vertex can change from being in a high component to a low component
at most once over the entire phase, the total cost per phase
is $\Ot(m\Delta)$, which is absorbed by the preprocessing cost.

\paragraph{Finale.}
The overall amortized cost per update
operation is $\Ot(\Delta^2+m/\Delta)$.
Set $\Delta=m^{1/3}$.
\end{proof}

Note that edge insertions and deletions in $G$ can be accomodated
easily (e.g., see Lemma~\ref{conn-deg} of the next section).

% [TC] space is not linear, right?  an issue to discuss, for
% the journal version

\section{Dynamic Geometric Connectivity with Sublinear Update Time}
\label{sec:geom}

In this section, we investigate geometric connectivity
problems: maintaining a set $S$ of $n$ objects, under insertions and
deletions of objects, so that we can decide whether two query
objects are connected in the intersection graph of~$S$.  
(In particular, we can
decide whether two query points are connected in the union of $S$
by finding two objects containing the two points, via range
searching, and testing connectedness for these two objects.)

By the biclique-cover technique from \cite{chan02subgraph},
the result from the previous section immediately implies
a dynamic connectivity method for
axis-parallel boxes with $\Ot(n^{2/3})$
update time and $\Ot(n^{1/3})$ query time
in any fixed dimension.

Unfortunately, this technique
is not strong enough to lead to sublinear results for other objects,
as we have explained in the introduction.
This is because (i)~the size of the maintained graph, $m=O(S(n)+nT(n))$,
may be too large and (ii)~the number of vertex updates 
triggered by an object update, $O(S(n)/n + T(n))$,
may be too large.

We can overcome the first obstacle by using a different strategy
that rebuilds the graph more often to keep it sparse; this is not
obvious and will be described precisely later during the proof of
Theorem~\ref{geomconn}.  The second obstacle is even more critical:
here, the key is to observe
that although each geometric update requires multiple vertex updates,
{\em many\/} of these vertex updates involves vertices of 
{\em low\/} degrees.  

\subsection{A degree-sensitive version of subgraph connectivity}

The first ingredient we need is a dynamic subgraph connectivity method that
works faster when the degree of the updated vertex is small.  
Fortunately, we can prove the following lemma, which
extends Theorem~\ref{conn} (if we set $\Delta=n^{1/3}$).  The method
follows that of Theorem~\ref{conn}, but with an extra
twist: not only do we classify components of $P$ as high or low,
but we also classify vertices of $Q$ as high or low.  

\begin{lemma}\label{conn-deg}
Let $1\le \Delta \le n$.
We can design a data structure for dynamic subgraph connectivity
for a graph $G=(V,E)$ with $m$ edges, having amortized vertex update time
\[ \Ot(\Delta^2 + \min\{m/\Delta, \deg(u)\}) \]
for a vertex $u$, query time $\Ot(\Delta)$,
preprocessing time $\Ot(m\Delta)$,
and amortized edge update time $\Ot(\Delta^2)$.
\end{lemma}
\begin{proof}
The data structure 
is the same as in the proof of Theorem~\ref{conn},
except for one difference: the definition of the graph $G^*$.

Call a vertex {\em high\/} if its degree exceeds $m/\Delta$, and {\em
low\/} otherwise.  Clearly, there are at most $O(\Delta)$ high vertices.
\begin{myitemize}
\item
We define a graph $G^*$ whose vertices are the vertices of $Q$
and components of $P$:
\begin{enumerate}
\item[(a$'$)] For each high vertex $u\in Q$ and each vertex $v\in Q$,
if $C[u,v]>0$, then create an edge $uv\in G^*$.
\item[(b)] For each vertex $u\in Q$ and high component $\gamma$ in $P$,
if $u\gamma\in\Gamma$, then create an edge $u\gamma\in G^*$.
\item[(b$'$)] For each low vertex $u\in Q$ and each component $\gamma$ in $P$,
if $u\gamma\in\Gamma$, then create an edge $u\gamma\in G^*$.
\item[(c)]  For each $u,v\in Q$, if $uv\in E$, then
create an edge $uv\in G^*$.
\end{enumerate}
We maintain $G^*$ in a data structure
for dynamic connectivity with polylogarithmic-time edge updates.
\end{myitemize}

\paragraph{Justification.}
We claim that two vertices of $Q$ are connected in the subgraph induced
by the active vertices in $G$ iff they
are connected in $G^*$.  The ``if'' direction is obvious.
For the ``only if'' direction, suppose two vertices $u,v\in Q$ are 
``directly'' connected in $G$ by being adjacent to a common 
component $\gamma$ in $P$.  If $\gamma$ is high, then edges
of type~(b) ensure that $u$ and $v$ are connected in $G^*$.  If
$u$ and $v$ are both low, then edges of type (b$'$) ensure that $u$ and $v$
are connected in $G^*$.  In the remaining case, at least one of
the two vertices, say, $u$ is high, and $\gamma$ is low; here, edges
of type~(a$'$) ensure that $u$ and $v$ are again connected in $G^*$.
The claim follows by concatenation.

\paragraph{Queries.}  
Given two vertices $v_1$ and $v_2$, if both are in $Q$, 
we can simply test whether they are connected in $G^*$.
If instead $v_j\ (j\in\{1,2\})$ is in a component $\gamma_j$, then we
can replace $v_j$ with any vertex of $Q$ adjacent to $\gamma_j$ in $G^*$.
If no such vertex exists, then because of type-(b$'$) edges,
$\gamma_j$ can only be adjacent to high vertices of $Q$.
We can exhaustively search for a high vertex in $Q$ adjacent to 
$\gamma_j$ in $\Gamma$, in $\Ot(\Delta)$ time,
and replace $v_j$ with such a vertex.  
If no such vertex exists, then $\gamma_j$ is an isolated component
and we can simply test whether $v_1$ and $v_2$ are both in $\gamma_j$.
The cost is $\Ot(\Delta)$.

\paragraph{Preprocessing per phase.}  
At the beginning of each phase, the cost to preprocess
the data structure is $\Ot(m\Delta)$ as before.
We can charge every update operation
with an amortized cost of $\Ot(m\Delta/q) = \Ot(\Delta^2)$.

\paragraph{Update of a high vertex $u$ in $Q$.}  
We need to update $O(q)$ edges of
types~(a$'$) and (c), and $O(m/\Delta)$ edges of type~(b) in $G^*$.
The cost is $\Ot(q + m/\Delta) = \Ot(m/\Delta)$.

\paragraph{Update of a low vertex $u$ in $Q$.}  
We need to update $O(\Delta)$ edges of
type~(a$'$), and  $O(\deg(u))$ edges of types (b), (b$'$), and (c) in $G^*$.
The cost is $\Ot(\deg(u)+\Delta)$.

\paragraph{Deletion of a vertex from a low/high component $\gamma$ in $P$.}
Proceeding exactly as in the proof of Theorem~\ref{conn}, 
we can update the data structure with
amortized cost $\Ot(\Delta^2)$.

\paragraph{Edge updates.}  
We can simulate the insertion of an edge $uv$ by
inserting a new low vertex $z$ adjacent to only $u$ and $v$ to $Q$.
Since the degree is 2, the cost is $\Ot(1)$.  We can later
simulate the deletion of this edge by deleting the vertex $z$ from $Q$.
\end{proof}

\subsection{Range searching tools from geometry}

Next, we need known range searching techniques.  These techniques give
linear-space data structures ($S(n)=\Ot(n)$) that can retrieve all objects
intersecting a query object in sublinear time ($T(n)=\Ot(n^{1-b})$)
for many types of geometric objects.  We assume that 
our class of geometric objects satisfies the following 
property for some constant $b>0$---this 
property neatly summarizes all we need to know from geometry.

\begin{property}\label{prop:geom}
Given a set $P$ of $n$ objects, we can form a collection $\calC$
of {\em canonical subsets\/} of total size $\Ot(n)$,  in $\Ot(n)$ time, such that
the subset of all objects of $P$ intersecting any object $z$ 
can be expressed as the union of disjoint subsets in a subcollection $\calC_z$
of $\Ot(n^{1-b})$ canonical subsets, in $\Ot(n^{1-b})$ time.
Furthermore, for every $1\le \Delta\le n$,
\begin{enumerate}
\item[\rm (i)] the number of subsets in $\calC_z$ 
of size exceeding $n/\Delta$ is $\Ot(\Delta^{1-b})$.
\item[\rm (ii)] the total size of all subsets in $\calC_z$ of 
size at most $n/\Delta$ is $\Ot(n/\Delta^b)$.
\end{enumerate}
\end{property}

The property is typically proved by applying a suitable ``partition
theorem'' in a recursive manner, thereby forming a so-called
``partition tree''; for example, see the work by Matou\v
sek~\cite{matousek92partition} or the survey by Agarwal and
Erickson~\cite{agarwal-erickson}.  Each canonical subset corresponds to a
node of the partition tree (more precisely, the subset of all objects
stored at the leaves underneath the node).  Matou\v sek's results
imply that $b=1/d-\eps$ is
attainable for simplices or constant-size polyhedra in $\R^d$.
(To go from simplex range searching to intersection searching,
one uses multi-level partition trees; e.g., see 
\cite{matousek93hierarchical}.)
Further results by Agarwal and Matou\v sek~\cite{agarwal94range} yield
$b=1/(d+1)-\eps$ for balls in $\R^d$ and nontrivial values of $b$ for other
families of curved objects (semialgebraic sets of constant degree).  
The special case of axis-parallel boxes
corresponds to $b=1$.  

The specific bounds in (i) and (ii) may not be
too well known, but they follow from the hierarchical way in which
canonical subsets are constructed.  For example, (ii) follows since
the subsets in $\calC_z$ of size at most
$n/\Delta$ are contained in $\Ot(\Delta^{1-b})$ subsets of size 
$\Ot(n/\Delta)$.  In fact, (multi-level)
partition trees guarantee a stronger inequality, 
$\sum_{C\in\calC_z}|C|^{1-b}=\Ot(n^{1-b})$, from which
both (i) and (ii) can be obtained after a moment's thought.

As an illustration, we can use the above property to develop a data
structure for a special case of dynamic geometric connectivity where
insertions are done in ``blocks'' but arbitrary deletions are to be supported.  
Although the insertion time is at
least linear, the result is good if the block size $s$ is sufficiently
large.  This subroutine will make up a part of the final solution.

\begin{lemma}\label{block}
We can maintain the connected components among
a set $S$ of objects in a data structure that
supports insertion of a block of $s$ objects
in $\Ot(n+sn^{1-b})$ amortized time ($s<n$), and 
deletion of a single object in $\Ot(1)$ amortized time.
\end{lemma}
\begin{proof}
We maintain a multigraph $H$ in a data structure for
dynamic connectivity with polylogarithmic edge update time
(which explicitly maintains the connected
components), where the vertices are the objects of $S$.  This multigraph
will obey the invariant that two objects are geometrically connected
iff they are connected in $S$.  We do not insist that $H$ 
has linear size.

\paragraph{Insertion of a block $B$ to $S$.}
We first form a collection $\calC$ of canonical subsets for $S\cup B$ by
Property~\ref{prop:geom}.  For each $z\in B$ and each $C\in\calC_z$, we {\em
assign\/} $z$ to $C$.  For each canonical subset $C\in\calC$, if $C$ is
assigned at least one object of $B$, then we create new edges in $H$
linking all objects of $C$ and all objects assigned to $C$ in a path.
(If this path overlaps with previous paths, we create multiple copies
of edges.)  The number of edges inserted is thus $\Ot(n+ |B|
n^{1-b})$.

\paragraph{Justification.}
The invariant is satisfied since all objects in a canonical subset $C$
intersect all objects assigned to $C$, and are thus all connected
if there is at least one object assigned to $C$.

\paragraph{Deletion of an object $z$ from $S$.}
For each canonical subset $C$ containing or assigned the object $z$,
we need to delete at most 2 edges and insert 1 edge to maintain the
path.  As soon as the path contains no object assigned to $C$, we
delete all the edges in the path.  Since the length of the path can
only decrease over the entire update sequence, the total number of
such edge updates is proportional to the initial length of the path.
We can charge the cost to edge insertions.
\end{proof}

\subsection{Putting it together}

We are finally ready to present our sublinear result for dynamic geometric
connectivity.  We again need the idea of
rebuilding periodically,
and splitting smaller sets from larger ones.  In addition to
the graph $H$ (of superlinear size) from Lemma~\ref{block},
which undergoes insertions only in blocks, 
the key lies in the 
definition of another subtly crafted intermediate graph $G$
(of linear size), maintained
this time by the subgraph connectivity structure of Lemma~\ref{conn-deg}.
The definition of this graph involves multiple types of vertices and edges.
The details of the analysis and the setting of parameters get more interesting.

\begin{theorem}\label{geomconn}
Assume $0<b\le 1/2$.
We can maintain a collection of objects in 
amortized update time $\Ot(n^{1-b^2/(2+b)})$ and
answer connectivity queries in time $\Ot(n^{b/(2+b)})$.
\end{theorem}
\begin{proof}
We divide the update sequence into phases, each consisting of $y :=
n^b$ updates.  The current objects are partitioned into two sets $X$
and $Y$, where $X$ undergoes only deletions and $Y$ undergoes both
insertions and deletions.  Each insertion is done to $Y$.  At the end
of each phase, we move the elements of $Y$ to $X$ and reset $Y$ to the
empty set.  This way, $|Y|$ is kept at most $y$ at all times.

At the beginning of each phase, we form a collection $\calC$
of canonical subsets for $X$ by Property~\ref{prop:geom}.

\paragraph{The data structure.} 
\begin{myitemize}
\item
We maintain the components of $X$ in the data structure from
Lemma~\ref{block}.
\item
We maintain the following graph $G$ for
dynamic subgraph connectivity, where the vertices are objects
of $X\cup Y$,
components of $X$, and the canonical subsets of the current phase:
\begin{enumerate}
\item[(a)] Create an edge in $G$ between each component of $X$ and
each of its objects.
\item[(b)] Create an edge in $G$ between each canonical subset
and each of its objects in $X$.
\item[(c)] Create an edge in $G$ between each object $z\in Y$ and
each canonical subset $C\in\calC_z$.  
Here, we {\em assign\/} $z$ to $C$.
\item[(d)] Create an edge in $G$ between 
every two intersecting objects in $Y$.
\item[(e)] We make a canonical subset active in $G$ iff it is assigned
at least one object in $Y$.  Vertices that are objects
or components are always active.  
\end{enumerate}
Note that there are $\Ot(n)$ edges of types (a) and (b), 
$\Ot(yn^{1-b})$ edges of type~(c), and $O(y^2)$ edges of type~(d).
For $y=n^b$, the size of $G$ is thus $\Ot(n)$.
\end{myitemize}

\paragraph{Justification.}
We claim that two objects are geometrically connected in $X\cup Y$ iff 
they are connected in the subgraph induced by the active vertices in 
the graph $G$. 
The ``only if'' direction is obvious.  For the ``if'' direction,
we note that all objects in an active canonical subset $C$ intersect
all objects assigned to $C$ and are thus all connected.

\paragraph{Queries.}
We answer a query by querying in the graph $G$.
The cost is $\Ot(\Delta)$.

\paragraph{Preprocessing per phase.}
Before a new phase begins, we need to update the components in $X$
as we move all elements of $Y$ to $X$ (a block insertion).  
By Lemma~\ref{block},
the cost is $\Ot(n+yn^{1-b})=\Ot(n)$.  
We can now 
reinitialize the graph $G$ containing $\Ot(n)$ edges of types (a)
and (b) in $\Ot(n\Delta)$ time by Lemma~\ref{conn-deg}. 
We can charge every update operation an amortized cost of
$\Ot(n\Delta/y)=\Ot(n^{1-b}\Delta)$.
 
\paragraph{Update of an object $z$ in $Y$.}
We need to update $\Ot(n^{1-b})$ edges of type (c) and
$O(y)$ edges of type (d) in $G$.  
The cost according to Lemma~\ref{conn-deg} is $\Ot(n^{1-b}\Delta^2)$.

Furthermore, because of (e),
we may have to update the status of as many as $\Ot(n^{1-b})$
vertices.  The number
of such vertices of degree exceeding $n/\Delta$ is $\Ot(\Delta^{1-b})$
by Property~\ref{prop:geom}(i),
and the total degree among such vertices of degree at most $n/\Delta$
is $\Ot(n/\Delta^b)$ by Property~\ref{prop:geom}(ii).
Thus, according to Lemma~\ref{conn-deg},
the cost of these vertex updates is 
$\Ot(n^{1-b}\Delta^2 + \Delta^{1-b}\cdot n/\Delta + n/\Delta^b)
=\Ot(n^{1-b}\Delta^2 + n/\Delta^b)$.

\paragraph{Deletion of an object $z$ in $X$.}
We first update the components of $X$.  By Lemma~\ref{block},
the amortized cost is $\Ot(1)$.  We can now update
the edges of type~(a) in $G$.  The total number of such edge updates
per phase is $O(n\lg n)$, by always splitting smaller components
from larger ones.  The amortized number of edge updates is
thus $\Ot(n/y)$.  The amortized cost is $\Ot((n/y)\Delta^2)=\Ot(n^{1-b}\Delta^2)$.

\paragraph{Finale.} 
The overall amortized cost per update operation is
$\Ot(n^{1-b}\Delta^2 + n/\Delta^b)$.   Set $\Delta=n^{b/(2+b)}$.
\end{proof}

Note that we can still prove the theorem for $b>1/2$, by handling the
$O(y^2)$ intersections among $Y$ (the type~(d) edges) in a less naive way.
However, we are not aware of any specific applications with $b\in (1/2,1)$.

\section{Offline Dynamic Geometric Connectivity}  \label{app:offline}

For the special case of offline updates, we can improve the result of
Section~\ref{sec:geom} for small values of $b$ by a different method
using rectangular matrix multiplication.

Let $M[n_1,n_2,n_3]$ represent the
cost of multiplying a Boolean $n_1\times n_2$ matrix $A$ with
a Boolean $n_2\times n_3$ matrix $B$.
Let $M[n_1,n_2,n_3\,|\,m_1,m_2]$ represent the same
cost under the knowledge that the number of 1's in $A$ is $m_1$
and the number of 1's in $B$ is $m_2$.
We can reinterpret this task in graph terms:
Suppose we are given a tripartite graph
with vertex classes $V_1,V_2,V_3$ of sizes $n_1,n_2,n_3$ respectively
where there are $m_1$ edges between $V_1$ and $V_2$ and 
$m_2$ edges between $V_2$ and $V_3$.  Then $M[n_1,n_2,n_3\,|\,m_1,m_2]$
represent the cost of deciding, for each $u\in V_1$
and $v\in V_3$, whether $u$ and $v$ are adjacent
to a common vertex in $V_2$.

\subsection{An offline degree-sensitive version of subgraph
connectivity}

We begin with an offline variant of Lemma~\ref{conn-deg}:

\begin{lemma}\label{conn-off}
Let $1\le \Delta\le q\le m$.
We can design a data structure for offline dynamic subgraph connectivity
for a graph $G=(V,E)$ with $m$ edges and $n$ vertices, under
the assumption that $O(\Delta)$ vertices
are classified as {\em high\/} and at most $m_H$ edges are incident
to high vertices.  
Updates of a low vertex $u$ take amortized time
\[ \Ot(M[\Delta,n,q\,|\,m_H,m]/q ~+~ \deg(u)),\]
updates of high vertices take amortized time $\Ot(q)$,
queries take time $\Ot(\Delta)$,
and preprocessing takes time $O(M[\Delta,n,q\,|\,m_H,m])$.
\end{lemma}
\begin{proof}
We divide the update sequence into phases, each consisting of $q$
low-vertex updates.  
The active vertices are partitioned into two sets $P$ and $Q$,
with $Q\subseteq Q_0$,
where $P$ and $Q_0$ are static and $Q$ undergoes both insertions
and deletions.
Each vertex insertion/deletion is done to $Q$.  At the end of each phase,
we reset $Q_0$ to hold all $O(\Delta)$ high vertices plus
the low vertices involved in the updates of the next phase,
reset $P$ to hold all active vertices not in $Q_0$, and 
reset $Q$ to hold all active vertices in $Q_0$.  
Clearly, $|Q|\le |Q_0|=O(q)$.

The data structure is the same as the one in the proof of
Lemma~\ref{conn-deg}, with one key difference: we only maintain the
value $C[u,v]$ when $u$ is a high vertex in $Q_0$ and $v$
is a (high or low) vertex in $Q_0$.  Moreover, we do not need to
distinguish between high and low components, i.e., all components
are considered low.

During preprocessing of each phase, we can now
compute $C[\cdot,\cdot]$ by matrix multiplication
in time $O(M[\Delta,n,q\,|\,m_H, m])$, since there are $O(\Delta)$
choices for the high vertex $u$ and $O(q)$ choices for
the vertex $v\in Q_0$.  The amortized cost per low-vertex update for
this step is $O(M[\Delta,n,q\,|\,m_H, m]/q)$.

Updating a high vertex $u$ in $Q$ now requires updating $O(q)$ edges of
types~(a$'$) and (c) (there are no edges of type~(b) now).
The cost is $\Ot(q)$.

Updating a low vertex $u$ in $Q$ requires updating $O(\Delta)$ edges of
type~(a$'$), and $O(\deg(u))$ edges of types (b$'$) and (c) in $G^*$.
The cost is $\Ot(\deg(u))$.

Deletions in $P$ do not occur now.
\end{proof}

\subsection{Sparse and dense rectangular matrix multiplication}

Sparse matrix multiplication can be reduced to multiplying
smaller dense matrices, by using a 
``high-low'' trick~\cite{alon97triangles}.  Fact~\ref{mult}(i) below can 
be viewed as
a variant of \cite[Lemma~3.1]{chan02subgraph} and a result
of Yuster and Zwick~\cite{yuster04FMM}---incidentally, this fact is 
sufficiently powerful
to yield a simple(r) proof of Yuster and Zwick's sparse
matrix multiplication result, when
combined with known bounds on dense rectangular matrix
multiplication.  Fact~\ref{mult}(ii) below states one known
bound on dense rectangular matrix multiplication which we will use.

\begin{fact}\label{mult}\
\begin{myitemize}
\item[\rm (i)] For $1\le t\le m_1$,
we have $M[n_1,n_2,n_3\,|\,m_1,m_2] = O(M[n_1,m_1/t,n_3] + m_2t)$.
\item[\rm (ii)] 
Let $\alpha=0.294$.  If $n_1\le\min\{n_2,n_3\}^\alpha$, then
$M[n_1,n_2,n_3]=\Ot(n_2n_3)$.
\end{myitemize}
\end{fact}
\begin{proof}
For (i), consider 
the tripartite graph setting with vertex classes $V_1,V_2,V_3$.
Call a vertex in $V_2$ {\em high\/} if it is incident to at least
$t$ vertices in $V_1$, and {\em low\/} otherwise.  There are at
most $O(m_1/t)$ high vertices.  For each $u\in V_1$ and $v\in V_3$,
we can determine whether $u$ and $v$ are adjacent to a common
high vertex, in $O(M[n_1,m_1/t,n_3])$ total time.  On the other hand,
we can enumerate all $(u,v)\in V_1\times V_3$ such that
$u$ and $v$ are adjacent to a common
low vertex, in $O(m_2t)$ time, by examining each edge $wv$ with
$(w,v)\in V_2\times V_3$ and each of the at most $t$ neighbors 
$u\in V_1$ of $w$.

For (ii),
Huang and Pan~\cite{huang98FMM} have shown that 
$M[n^\alpha,n,n]=M[n,n^\alpha,n]=\Ot(n^2)$.  Thus,
\[ M[n_1,n_2,n_3] ~=~ O(\up{n_2/n_1^{1/\alpha}}\cdot \up{n_3/n_1^{1/\alpha}}
\cdot M[n_1,n_1^{1/\alpha},n_1^{1/\alpha}])
~=~ O(\up{n_2/n_1^{1/\alpha}}\cdot \up{n_3/n_1^{1/\alpha}}\cdot
        n_1^{2/\alpha}).
\]

\vspace{-\bigskipamount}
\end{proof}

\subsection{Putting it together}

We now present our offline result for dynamic geometric connectivity
using Lemma~\ref{conn-off}.  Although we also use Property~\ref{prop:geom},
the design of the key graph $G$ is quite
different from the one in the proof of
Theorem~\ref{geomconn}.  For instance, the size of the graph is larger
(and no longer $\Ot(n)$), but the number of edges incident to high
vertices remains linear; furthermore, each object update triggers
only a constant number of vertex updates in the graph.  All the details 
come together in the analysis to lead to some intriguing
choices of parameters.

\begin{theorem}\label{geomconn-off}
Assume $0<b\le 1$.  Let $\alpha=0.294$.
We can maintain a collection of objects in 
amortized time $\Ot(n^{\frac{1+\alpha-b\alpha}{1+\alpha-b\alpha/2}})$ 
for offline updates and answer connectivity queries in time 
$\Ot(n^{\frac{\alpha}{1+\alpha-b\alpha/2} })$.
\end{theorem}
\begin{proof}
We divide the update sequence into phases, each consisting of
$q$ updates, where $q$ is a parameter satisfying
$\Delta \le q\le n/\Delta^{1-b}$.
The current objects are partitioned into two sets $X$ and $Y$,
with $Y\subseteq Y_0$
where $X$ and $Y_0$ are static and $Y$ undergoes both insertions
and deletions. 
Each insertion/deletion is done to $Y$.  At the end of each phase,
we reset $Y_0$ to hold all objects involved the objects of the next phase,
$X$ to hold all current objects not in $Y_0$, and $Y$ to hold all current
objects in $Y_0$.  Clearly, $|Y|\le |Y_0|=O(q)$.

At the beginning of each phase, we form a collection $\calC$
of canonical subsets for $X\cup Y_0$ by Property~\ref{prop:geom}.

\paragraph{The data structure.}
\begin{myitemize}
\item
We maintain the components of $X$ in the data structure from
Lemma~\ref{block}.
\item
We maintain the following graph $G$ for offline dynamic 
subgraph connectivity, where the vertices are objects of $X\cup Y_0$,
components of $X$, and canonical subsets of size exceeding $n/\Delta$:
\begin{enumerate}
\item[(a)] Create an edge in $G$ between each component of $X$ and
each of its objects.
\item[(b)] Create an edge in $G$ between each canonical subset $C$ of
size exceeding $n/\Delta$ and each of its objects in $X\cup Y$.
\item[(c)] Create an edge in $G$ between each object $z\in Y_0$ and
each canonical subset $C\in \calC_z$ of size exceeding $n/\Delta$.   
Here, we {\em assign\/} $z$ to $C$.
\item[(d)] Create an edge in $G$ between each object $z\in Y_0$
and each object in the union of the canonical subsets in $\calC_z$
of size at most $n/\Delta$.
\item[(e)] We make a canonical subset active in $G$ iff it is assigned
at least one object in $Y$.  
We make the vertices in $X\cup Y$ active, and all components active.
The {\em high\/} vertices are precisely the
canonical subsets of size exceeding $n/\Delta$; there are $\Ot(\Delta)$ such vertices.
\end{enumerate}
Note that there are $\Ot(n)$ edges of types (a) and (b),
$\Ot(q\Delta^{1-b})$ edges of type~(c) by Property~\ref{prop:geom}(i), and
$\Ot(qn/\Delta^b)$ edges of type~(d) by Property~\ref{prop:geom}(ii).
So the graph has size $m = \Ot(n + qn/\Delta^b) = \Ot(qn/\Delta^b)$, 
and the number of edges incident
to high vertices is $m_H = \Ot(n + q\Delta^{1-b}) = \Ot(n)$.
\end{myitemize}

\paragraph{Preprocessing per phase.}
Before a new phase begins, we need to update the components in $X$
as we delete $O(q)$ vertices from and insert $O(q)$ vertices to $X$.
By Lemma~\ref{block}, the cost is $\Ot(n + qn^{1-b})$.
We can then determine the edges of type (a)
in $G$ in $\Ot(n)$ time.  We can now initialize $G$
in $O(M[\Delta,n,q\,|\,n,m])$ time by Lemma~\ref{geomconn-off}.
We can charge every update operation with an amortized
cost of $\Ot(M[\Delta,n,q\,|\,n,m]/q + n/q + n^{1-b})$.

\paragraph{Update of an object $z$ in $Y$.}
We need to make a single vertex update $z$ in $G$, which has
degree $\Ot(n/\Delta^b)$ by Property~\ref{prop:geom}(ii).
Furthermore, we may have to change the status of as many
as $\Ot(\Delta^{1-b})$ high vertices by Property~\ref{prop:geom}(i).
According to Lemma~\ref{geomconn-off},
the cost of these vertex updates is 
$\Ot(M[\Delta,n,q\,|\,n,m]/q + n/\Delta^b + \Delta^{1-b}q)$.

\paragraph{Finale.}
By Fact \ref{mult}, assuming that
$\Delta\le q^\alpha$ and $q\le n/t$, we have
$M[\Delta,n,q\,|\,n,m]\ =\ O(M[\Delta,n/t,q] + mt)
\ =\ \Ot(nq/t + nqt/\Delta^b)$.
Choosing $t=\Delta^{b/2}$ gives $\Ot(nq/\Delta^{b/2})$.

The overall amortized cost per update operation is thus
$\Ot(n/\Delta^{b/2} + \Delta^{1-b}q + n/q + n^{1-b}).$
Set $\Delta=q^\alpha$ and $q=n^{\frac{1}{1+\alpha-b\alpha/2}}$ and
the result follows.  (Note that indeed 
$\Delta \le q\le n/\Delta^{1-b}$ and $q\le n/t$ for these choices of parameters.)
\end{proof}

Compared to Theorem~\ref{geomconn}, the dependence on $b$ of the
exponent in the update bound is only $1-\Theta(b)$ rather than 
$1-\Theta(b^2)$.
The bound is better, for example, for $b\le 1/4$.  
%However, the use of
%fast matrix multiplication makes the method less practical.

\section{Open Problems}

Our work opens up many interesting directions for further
research.  For subgraph connectivity, an obvious question is
whether the $\Ot(m^{2/3})$ vertex-update bound
can be improved (without or with FMM);
as we have mentioned, improvements beyond $\sqrt{m}$ without FMM are
not possible without a breakthrough on the triangle-finding problem.  
An intriguing question is whether for dense graphs we can
achieve update time sublinear in $n$, i.e., $O(n^{1-\eps})$ 
(or possibly even sublinear in the degree).  

For geometric connectivity, it would be desirable to determine
the best update bounds for specific shapes such as line
segments and disks in two dimensions.
Also, {\em directed\/} settings of geometric connectivity
arise in applications and are worth studying;
for example, when sensors' transmission ranges are balls
of different radii or wedges, a sensor may lie in another sensor's
range without the reverse being true.

For both subgraph and geometric connectivity, we can reduce the 
query time at the expense of increasing the update
time, but we do not know whether constant or polylogarithmic query time 
is possible with sublinear update time in general 
(see \cite{afshani06connectivity} for a result
on the 2-dimensional orthogonal
special case).  
Currently, we do not know how to obtain our update bounds with
linear space (e.g., Theorem~\ref{conn} requires $\Ot(m^{4/3})$ space), nor 
do we know how to get good
worst-case update bounds (since the known polylogarithmic results
for connectivity under edge updates are all amortized).
Also, the queries we have
considered are about connectivity between two vertices/objects.
Can nontrivial results be obtained for richer queries such as 
counting the number of connected components 
(see \cite{afshani06connectivity} on the 2-dimensional orthogonal
case), or perhaps shortest paths or minimum cut?

%* improve update time (but Timothy's reduction says no more than
%  $\sqrt{m}$).

%* obtain something sublinear even for dense graphs (maybe always beat
%  degree!)

%* query time can be made constant?

%* best geometric conn in natural shapes (balls, 2D things).

%* geometric connectivity in directed settings (sensors with balls of
%  different radii, wedges ~~ receive from any angle, send only in wedge)

%* support richer queries like count conn comp, maybe MST, SPs, min cut...

{%\small
\bibliographystyle{plain} %alpha}
%\bibliography{../../general}

}

\end{document}